\documentclass[10pt]{iopart}
\usepackage{iopams}
\usepackage{amsgen,amsfonts,amsbsy}
\usepackage{amssymb,amsthm,amscd}
\usepackage{amstext,latexsym}
\usepackage{imath}
\usepackage{graphicx}
\pdfoutput=1

\theoremstyle{plain}
\newtheorem{theorem}{Theorem}[section]
\newtheorem{proposition}[theorem]{Proposition}
\newtheorem{lemma}[theorem]{Lemma}
\newtheorem{corollary}[theorem]{Corollary}

\theoremstyle{definition}
\newtheorem{definition}[theorem]{Definition}

\newtheorem{remark}[theorem]{Remark}

\newcommand{\Z}{\mathbb{Z}}
\newcommand{\N}{\mathbb{N}}

\newcommand{\C}{\mathbb{C}}
\newcommand{\set}[2]{\{#1|\ #2\}}
\newcommand{\sub}{\subseteq}
\newcommand{\gen}[1]{\left\langle #1\right\rangle}
\newcommand{\M}{\mathrm{M}}
\newcommand{\U}{\mathrm{U}}
\newcommand{\GL}{\mathrm{GL}}
\newcommand{\Ad}{\mathrm{Ad}}
\newcommand{\Inn}{\mathrm{Int}}
\newcommand{\lcm}{\mathrm{lcm}}
\newcommand{\Aut}{\mathrm{Aut}}
\newcommand{\End}{\mathrm{End}}
\newcommand{\diag}{\mathrm{diag}}

\newcommand{\SL}{\mathrm{SL}}
\newcommand{\Sp}{\mathrm{Sp}}

\renewcommand{\sl}{\mathrm{sl}}
\renewcommand{\P}{\mathcal{P}}

\newenvironment{smat2}[4]{\begin{smallmatrix}#1&#2\\#3&#4\end{smallmatrix}}{}

\begin{document}

\title[Symmetries of finite Heisenberg groups]{Symmetries
of finite Heisenberg groups for multipartite systems}
\author{M. Korbel\'{a}\v{r} $^1$ and J. Tolar $^2$}
\address{$^1$ Department of Mathematics and Statistics \\
 Faculty of Science, Masaryk University\\
Kotl\'{a}\v{r}sk\'{a} 2, 611 37 Brno, Czech Republic \\
$^2$ Department of Physics \\
 Faculty of Nuclear Sciences and  Physical  Engineering\\
 Czech Technical University in Prague\\ B\v{r}ehov\'{a} 7,
 115 19 Prague 1, Czech Republic}
 \eads{\mailto{miroslav.korbelar@gmail.com},
 \mailto{jiri.tolar@fjfi.cvut.cz}}



\begin{abstract}
A composite quantum system comprising a finite number $k$ of
subsystems which are described with position and momentum variables
in $\Z_{n_{i}}$, $i=1,\dots,k$, is considered. Its Hilbert space is
given by a $k$-fold tensor product of Hilbert spaces of dimensions
$n_{1},\dots,n_{k}$. Symmetry group of the respective finite
Heisenberg group is given by the quotient group of certain
normalizer. This paper extends our previous investigation of
bipartite quantum systems to arbitrary multipartite systems of the
above type. It provides detailed description of the normalizers and
the corresponding symmetry groups. The new class of symmetry groups
represents a very specific generalization of symplectic groups over
modular rings. As an application, a new proof of existence of the
maximal set of mutually unbiased bases in Hilbert spaces of prime
power dimensions is provided.
\end{abstract}

\section{Introduction}

The Heisenberg Lie algebra and the Heisenberg-Weyl group lie at the
heart of quantum mechanics \cite{Weyl}. Therefore their symmetries
induced by unitary automorphisms play very important role in quantum
kinematics as well as quantum dynamics. The growing interest in
quantum communication science has pushed the study of quantum
systems with finite-dimensional Hilbert spaces to the forefront,
both single systems and composite systems. For them the finite
Heisenberg groups provide the basic quantum observables. It is then
clear that the symmetries of finite Heisenberg groups uncover deeper
structure of finite-dimensional quantum mechanics.

Our continuing interest in finite-dimensional quantum mechanics goes
back to the paper \cite{StovTolar84} where finite-dimensional
quantum mechanics was developed as quantum mechanics on
configuration spaces given by finite sets equipped with the
structure of a finite Abelian group. In our recent paper
\cite{normalizer} detailed characterization was given of the
symmetry groups of finite Heisenberg groups for composite quantum
systems consisting of two subsystems with arbitrary dimensions $n$,
$m$. In this contribution these results for bipartite systems are
extended to the general finitely composed systems consisting of an
arbitrary number $k$ of subsystems with arbitrary dimensions
$n_{1},\dots,n_{k}$. Their Hilbert spaces are given by $k$-fold
tensor products of Hilbert spaces of dimensions $n_{1},\dots,n_{k}$.

In the course of work it turned out that --- even if the idea of the
present paper is similar to \cite{normalizer} --- intermediate steps
could not be taken over literally from \cite{normalizer}, but had to
be carefully developed in the general multipartite situation.
Preliminary results were given in \cite{normalizerJPCS}.

The exposition starts with introductory material on
finite-dimensional quantum mechanics in section \ref{sec2}. In
section \ref{sec3} the group $\Sp_{[n_{1},\dots,n_{k}]}$ is
introduced. Its further properties are given in section \ref{sec4}.
The proof that $\Sp_{[n_{1},\dots,n_{k}]}$ is indeed the symmetry
group is contained in section \ref{sec5}, theorem \ref{37}. The
reader will see that the new family of finite groups deserves the
chosen notation because they generalize finite symplectic groups
over modular rings \cite{symplectic}.

A special subclass $\Sp_{2k}(Z_p)$ of our family of symmetry groups
is applied in section \ref{sec6} to an alternative proof of
existence of the maximal set of mutually unbiased bases in Hilbert
spaces of prime power dimensions $p^k$
\cite{WoottersFields89,Bandyo}. We remind that previous group
theoretical construction of mutually unbiased bases presented in
\cite{SulcTolar07} was based on the symmetry groups $\SL_{2}(\Z_p)$
of the finite Heisenberg groups for Hilbert spaces of prime
dimensions $p$.

\section{Finite-dimensional quantum mechanics}\label{sec2}

For reader's convenience we very briefly repeat the basic notions of
finite-dimensional quantum mechanics (FDQM) \cite{Vourdas,Kibler}
for a single-component system with the Hilbert space
$\mathcal{\ell}^{2}(\Z_n)$ of arbitrary dimension $n\in\N$. In this
case the cyclic group $\Z_n$ serves as the underlying configuration
space.

We follow the notation of \cite{normalizer}, where further details
can be found. For a given $n\in\N$ we set
 $$\omega_{n}:=e^{2\pi \mathrm{i}/n}\in\C.$$
Let $Q_{n}$ and $P_{n}$ denote the \textit{generalized Pauli
matrices} of order $n$,
$$Q_{n}:=\diag(1,\omega_{n},\omega_{n}^{2},\dots,\omega_{n}^{n-1})\in
\GL_{n}(\C)$$
 and
$$P_{n}\in \GL_{n}(\C),\quad \text{where}\quad (P_{n})_{i,j}:=\delta_{i,j-1}, \quad
i,j\in\Z_n.$$
 They belong to the group $\U_n(\C)$ of $n\times n$ unitary matrices
in $\mathcal{\ell}^{2}(\Z_n)$. Let $I_{n}$ denote the $n\times n$
unit matrix. The subgroup of unitary matrices in $\GL_{n}(\C)$
generated by $Q_{n}$ and $P_{n}$,
$$\Pi_{n}:=\{ \omega_{n}^j Q_{n}^k P_{n}^l \vert
 j,k,l\in\{0,1,\dots,n-1 \} \}$$
 is called the {\it finite Heisenberg group}.
Recall that the order of $\Pi_{n}$ is $n^3$, its center is
$Z(\Pi_{n})=\{\omega_{n}^j I_n \vert  j\in\{0,1,\dots,n-1 \} \}$ and
\begin{equation} \label{2PQ}
 P_{n}Q_{n}=\omega_{n}Q_{n}P_{n}.
\end{equation}

For $M\in\GL_{n}(\C)$ let $\Ad_{M}\in \Inn(\GL_{n}(\C))$ be the {\it
inner automorphism} of the group $\GL_{n}(\C)$ induced by operator
$M\in \GL_{n}(\C)$, i.e.
 $$\Ad_{M}(X)=MXM^{-1}\quad \textrm{for}\quad X\in \GL_{n}(\C).$$

\begin{definition}\label{2a}
We define $\P_n$ as the group
 $$ \P_n=\{\Ad_{Q_n^i P_n^j} \vert (i,j)\in \Z_n \times \Z_n \}.$$
 It is an Abelian subgroup of $\Inn(\GL_{n}(\C))$ and is generated by
 two commuting automorphisms $\Ad_{Q_n}$, $\Ad_{P_n}$,
 $$ \P_n = \gen{\Ad_{Q_n},\Ad_{P_n}}.$$
\end{definition}
\noindent A geometric view is sometimes useful that $\P_n$ is
isomorphic to the \textit{quantum phase space} identified with the
Abelian group $\Z_n^2$ \cite{Wootters87,SulcTolar07}.

Let us recall the usual properties of the matrix tensor product
$\otimes$. Let $A,A'\in\GL_{n}(\C)$, $B,B'\in\GL_{m}(\C)$ and
$\alpha\in\C$. Then:
\begin{enumerate}
\item[(i)] $(A\otimes B)(A'\otimes B')=AA'\otimes BB'$.
\item[(ii)] $\alpha(A\otimes B)=(\alpha A)\otimes B=A\otimes(\alpha B)$.
\item[(iii)] $A\otimes B=I_{nm}$ if and only if
there is non-zero $\alpha\in\C$ such that $A=\alpha I_{n}$ and
$B=\alpha^{-1}I_{m}$.
\end{enumerate}

Finally we introduce the main notions for the $k$-partite situation
where the group $\Z_{n_{1}}\times\cdots\times\Z_{n_{k}}$ (with
$n_{1},\dots,n_{k}\in\N$) serves as the configuration space.
\begin{definition}\label{6}
 Let $n_{1}\dots n_{k}=N$. We define
$$\mathcal{P}_{(n_{1},\dots,n_{k})}=
\set{\Ad_{M_{1}\otimes\cdots\otimes M_{k}}}{M_{i}\in\Pi_{n_{i}}}
\sub\Inn(\GL_{N}(\C)).$$
 In the following we shall work with generating elements of the
 finite Heisenberg group
 $\Pi_{n_{1}}\otimes \dots \otimes \Pi_{n_{k}}$,
\begin{equation} \label{2A}
A_{2i-1}:=I_{n_{1}\cdots n_{i-1}}\otimes P_{n_{i}}\otimes
I_{n_{i+1}\cdots n_{k}}, \ \ \ A_{2i}:=I_{n_{1}\cdots
n_{i-1}}\otimes Q_{n_{i}}\otimes I_{n_{i+1}\cdots n_{k}},
 \end{equation} for
$i=1,\dots,k$ and the corresponding inner automorphisms
$$e_{j}:=\Ad_{A_{j}} \quad \textrm{for}\quad  j=1,\dots,2k.$$
\end{definition}
\noindent Clearly, $\mathcal{P}_{(n_{1},\dots,n_{k})}$ is an Abelian
group given by the direct product of the groups $\gen{e_{j}}$, where
$j=1,\dots,2k$.

\begin{lemma}\label{10}
Let $n_{1} \dots n_{k}=N$. Then $\mathcal{P}_{(n_{1},\dots,n_{k})}$
is a maximal Abelian subgroup of diagonalizable automorphisms in
$\Inn(\GL_{N}(\C))$.
\end{lemma}
This subgroup has been called a MAD-group \cite{PZ89,HPPT02} and it
is the subgroup of $\Inn(\GL_{N}(\C))$ such that its centralizer in
$\Inn(\GL_{N}(\C))$ is equal to $\mathcal{P}_{(n_{1},\dots,n_{k})}$.
The proof for the bipartite case was given in
\cite[4.3]{normalizer}.

\section{The symmetry group $\Sp_{[n_{1},\dots,n_{k}]}$}\label{sec3}

In this section the group $\Sp_{[n_{1},\dots,n_{k}]}$ is defined as
a matrix subgroup of $\GL_{N}(\C)$ and its principal properties are
proved. It will be constructed in several steps. Through this
section and sections \ref{sec4} and \ref{sec5} let
$n_{1},\dots,n_{k}\in\N$ be fixed numbers and let
$\M_{n}(\mathcal{R})$ be the ring of $n \times n$ matrices with
entries from the ring $\mathcal{R}$.

Our construction starts with a set of block matrices:
\begin{definition}\label{11}
Let $\mathcal{M}_{[n_{1},\dots,n_{k}]}$ be a set consisting of
$k\times k$ matrices $H$ composed of $2\times 2$ blocks
  $$H_{ij}=\frac{n_{i}}{\gcd(n_{i},n_{j})}A_{ij}$$
where $A_{ij}\in\M_{2}(\Z_{n_{i}})$ for $i,j=1,\dots,k$, are
$2\times 2$ matrices over $\Z_{n_{i}}$.

It is useful to take such matrices over $\Z$,
$$\mathcal{S}_{[n_{1},\dots,n_{k}]}:=\Big\{H\in\M_{k}(\M_{2}(\Z)) |\
A_{ij}\in\M_{2}(\Z), H_{ij}=\frac{n_{i}}{\gcd(n_{i},n_{j})}A_{ij},
i,j=1,\dots,k \Big\},$$
 and, using a special diagonal matrix
 $$D:=\diag\Big(\frac{\lcm(n_{1},\dots,n_{k})}{n_{1}}I_{2},
 \dots,\frac{\lcm(n_{1},\dots,n_{k})}{n_{k}}I_{2}\Big)
 \in\mathcal{S}_{[n_{1},\dots,n_{k}]},$$
to define a congruence $\equiv$ on
$\mathcal{S}_{[n_{1},\dots,n_{k}]}$:
 $$H\equiv G\ \Leftrightarrow\
 DH\equiv_{\lcm(n_{1},\dots,n_{k})}DG, \quad \textrm{where}
\quad H,G\in\mathcal{S}_{[n_{1},\dots,n_{k}]}.
 $$
Further an adjoint $H^{\ast}\in\mathcal{S}_{[n_{1},\dots,n_{k}]}$ of
$H\in\mathcal{S}_{[n_{1},\dots,n_{k}]}$,
$H_{ij}=\frac{n_{i}}{\gcd(n_{i},n_{j})}A_{ij}$
 is defined by
$$(H^{\ast})_{ij}=\frac{n_{i}}{\gcd(n_{i},n_{j})}A_{ji}^{T}.$$
For convenience we put $\ell:=\lcm(n_{1},\dots,n_{k})$ in this
section.
\end{definition}

\begin{remark}\label{11e}\
The above definitions lead to the following properties of
$\mathcal{M}_{[n_{1},\dots,n_{k}]}$:
\begin{enumerate}
\item\label{a} Let $d,n,a,b\in\Z$ and $d\mid n$.
\ Then the congruence $\frac{n}{d}a\equiv_{n}\frac{n}{d}b$ is
equivalent to $a\equiv_{d}b$, i.e. $a=b \pmod{d}$.

\item\label{b} By (\ref{a}), we see that
$\mathcal{M}_{[n_{1},\dots,n_{k}]}=\mathcal{S}_{[n_{1},\dots,n_{k}]}/_{\equiv}$.

\item\label{c} Let $i,j,m\in\{1,\dots, k\}$.
Then $\frac{n_{i}}{\gcd(n_{i},n_{j})}\ \ \big|\ \
 \frac{n_{i}}{\gcd(n_{i},n_{m})}\frac{n_{m}}{\gcd(n_{m},n_{j})}$.

  \noindent Indeed, $\gcd(n_{m},n_{j})\cdot\gcd(n_{i},n_{m})$
 divides $n_{m}n_{i}$ and also $n_{j}n_{m}$.
 Hence $\gcd(n_{m},n_{j})\cdot\gcd(n_{i},n_{m})$
 divides $\gcd(n_{m}n_{i},n_{m}n_{j})=n_{m}\gcd(n_{i},n_{j})$ and
 thus
 $$\frac{n_{m}\gcd(n_{i},n_{j})}{\gcd(n_{i},n_{m})\gcd(n_{m},n_{j})}\in\Z.$$

\item\label{d} Using (\ref{c}) we get that
$\mathcal{S}_{[n_{1},\dots,n_{k}]}$ is a subring of
$\M_{k}(\M_{2}(\Z))$.

\item\label{e} It is easy to verify that $DH=(H^{\ast})^{T}D$
for every $H\in\mathcal{S}_{[n_{1},\dots,n_{k}]}$.

\item $\equiv$ is a ring congruence on $\mathcal{S}_{[n_{1},\dots,n_{k}]}$.
Thus, by (\ref{b}) and (\ref{d}), $\mathcal{M}_{[n_{1},\dots,n_{k}]}$ is
(with the usual matrix multiplication and addition) a ring.

  {\small  \noindent It is enough to show that $\mathcal{I}:=
\set{H\in\mathcal{S}_{[n_{1},\dots,n_{k}]}}{H\equiv 0}$ is an ideal
in $\mathcal{S}_{[n_{1},\dots,n_{k}]}$. Let
$H,G\in\mathcal{S}_{[n_{1},\dots,n_{k}]}$ and $H\in\mathcal I$. Then
$DH\equiv_{\ell}0$. Hence by (\ref{e}) we have
$D(GH)\equiv_{\ell}(G^{\ast})^{T}(DH)\equiv_{\ell}0$ and
$GH\in\mathcal{I}$. The rest is obvious.}

\item\label{f} $\mathcal{M}_{[n_{1},\dots,n_{k}]}$ has a natural action
(via the matrix multiplication) on the quantum phase space
$\Z_{n_{1}}^{2}\times\cdots\times\Z_{n_{k}}^{2}$.

  {\small \noindent  Clearly,
$\Z_{n_{1}}^{2}\times\cdots\times\Z_{n_{k}}^{2}$ can be viewed as
$\Z^{2k}$ factorized by the equivalence: $x\equiv y$ if and only if
$Dx\equiv_{\ell}Dy$, where $x,y\in\Z^{2k}$. One needs only to show
that $H\equiv G$ and $x\equiv y$ implies $Hx\equiv Gy$ for
$H,G\in\mathcal{S}_{[n_{1},\dots,n_{k}]}$ and $x,y\in\Z^{2k}$. Let
$DH\equiv_{\ell}DG$ and $Dx\equiv_{\ell}Dy$. Then
$DHx\equiv_{\ell}(DG)x\equiv_{\ell}(G^{\ast})^{T}(Dx)\equiv_{\ell}(G^{\ast})^{T}Dy\equiv_{\ell}DGy$
and thus $Hx\equiv Gy$.}

\item\label{g} $\mathcal{M}_{[n_{1},\dots,n_{k}]}$ is a finite set
of matrices closed under usual matrix multiplication and containing
the unit matrix as neutral element, i.e. it is a finite monoid.
\end{enumerate}
\end{remark}

Property (\ref{f}) can be naturally extended to any endomorphism of
the quantum phase space
$\mathcal{K}=\Z_{n_{1}}^{2}\times\cdots\times\Z_{n_{k}}^{2}$:

\begin{proposition}\label{31}
For every $\alpha\in
\End(\Z_{n_{1}}^{2}\times\cdots\times\Z_{n_{k}}^{2})$ there is a
unique $H\in\mathcal{M}_{[n_{1},\dots,n_{k}]}$ such that
$\alpha(x)=Hx$ for every
$x\in\Z_{n_{1}}^{2}\times\cdots\times\Z_{n_{k}}^{2}$.
 The map
$$\Phi:\End(\Z_{n_{1}}^{2}\times\cdots\times\Z_{n_{k}}^{2})\to
\mathcal{M}_{[n_{1},\dots,n_{k}]},$$ where $\Phi(\alpha):=H$ is a
ring isomorphism.
\end{proposition}
\begin{proof}
Let $\{f_{1},\dots,f_{2k}\}$ be the canonical generating set of
$\Z_{n_{1}}^{2}\times\cdots\times\Z_{n_{k}}^{2}$. For every
$\alpha\in \End(\Z_{n_{1}}^{2}\times\cdots\times\Z_{n_{k}}^{2})$
there are $h_{ij}\in\Z$ such that
$\alpha(f_{j})=\sum^{2k}_{i=1}h_{ij} f_{i}$. The order of $f_{2i-1}$
and $f_{2i}$ is $n_{i}$ for $i=1,\dots,k$. Hence we have
$1=\alpha(n_{i}f_{2i-1})=\sum^{2k}_{j=1}(n_{i}h_{j,2i-1})f_{j}$ and
$1=\alpha(n_{i}f_{2i})=\sum^{2k}_{j=1}(n_{i}h_{j,2i})f_{j}$. Thus
$n_{i}h_{2j-1,2i}\equiv_{n_{j}}0\equiv_{n_{j}}n_{i}h_{2j,2i}$ for
every $j=1,\dots,k$. It follows that
$\frac{n_{i}}{\gcd(n_{i},n_{j})}h_{2j-1,2i}\equiv_{n_{j}/\gcd(n_{i},n_{j})}0
\equiv_{n_{j}/\gcd(n_{i},n_{j})}\frac{n_{i}}{\gcd(n_{i},n_{j})}h_{2j,2i}$
and $h_{2j-1,2i},h_{2j,2i}\in\frac{n_{j}}{\gcd(n_{i},n_{j})}\Z$ for
every $j=1,\dots,k$. Now, consider $h_{ij}$ modulo $\lceil
n_{i}/2\rceil$. Put
$H=(h_{ij})_{i,j=1,\dots,2k}\in\mathcal{M}_{[n_{1},\dots,n_{k}]}$
and the rest is easy.
\end{proof}

\begin{remark}\label{122}
Properties of the adjoint operation given below mean that
$\mathcal{S}_{[n_{1},\dots,n_{k}]}$ and
$\mathcal{M}_{[n_{1},\dots,n_{k}]}$
 have the structure of a $\ast$-ring:
\begin{enumerate}
\item\label{a2} Let $H,G\in\mathcal{S}_{[n_{1},\dots,n_{k}]}$.
Then $(H^{\ast})^{\ast}=H$, $(H+G)^{\ast}=H^{\ast}+G^{\ast}$ and
$(HG)^{\ast}=G^{\ast}H^{\ast}$, i.e. the operation $\ast$ is an
involutive ring antihomomorphism.

{\small \noindent Let
$H_{ij}=\frac{n_{i}}{\gcd(n_{i},n_{j})}A_{ij}\in\Z_{n_{i}}$ and
$G_{ij}=\frac{n_{i}}{\gcd(n_{i},n_{j})}B_{ij}\in\Z_{n_{i}}$ for
$i,j=1,\dots,k$. Then
$(G^{\ast}H^{\ast})_{ij}=
\sum\limits^{k}_{m=1}\frac{n_{i}}{\gcd(n_{i},n_{m})}
\frac{n_{m}}{\gcd(n_{m},n_{j})}B_{mi}^{T}A_{jm}^{T}=$

$=\frac{n_{i}}{\gcd(n_{i},n_{j})}\sum\limits^{k}_{m=1}
\frac{n_{m}\gcd(n_{i},n_{j})}{\gcd(n_{i},n_{m})\gcd(n_{m},n_{j})}
(A_{jm}B_{mi})^{T}=(HG)^{\ast}_{ij}$. \\ The rest is obvious.}

\item\label{b2} Let $H,G\in\mathcal{S}_{[n_{1},\dots,n_{k}]}$.
Then $H\equiv G$ implies $H^{\ast}\equiv G^{\ast}$.
Thus the operation $\ast$ is well defined on $\mathcal{M}_{[n_{1},\dots,n_{k}]}$.

{\small \noindent Indeed, let $DH\equiv_{\ell}DG$. Then
$DH^{\ast}\equiv_{\ell}H^{T}D\equiv_{\ell}G^{T}D\equiv_{\ell}DG^{\ast}$
and $H^{\ast}\equiv G^{\ast}$.}
\end{enumerate}
\end{remark}

Now we are going to define $\Sp_{[n_{1},\dots,n_{k}]}$.
\begin{definition}\label{19}
Denote $
 J=\diag(J_{2},\dots,J_{2})\in\mathcal{M}_{[n_{1},\dots,n_{k}]}$
where $J_{2}=\left(\begin{array}{cc}0&1\\-1&0\end{array}\right)$
 and put
 \begin{equation} \label{group}
 \Sp_{[n_{1},\dots,n_{k}]}:=
 \set{H\in \mathcal{M}_{[n_{1},\dots,n_{k}]}}{\ H^{\ast}JH= J}.
 \end{equation}
\end{definition}

The following proposition implies that $\Sp_{[n_{1},\dots,n_{k}]}$
is a finite subgroup of the monoid
$\mathcal{M}_{[n_{1},\dots,n_{k}]}$.
 \begin{proposition}\label{19.1}
Let $\mathcal{M}$ be a finite monoid and $x\mapsto x^{\ast}$
 an involutive anti-homomorphism of $\mathcal{M}$
(i.e. $(x^{\ast})^{\ast}=x$ and $(xy)^{\ast}=y^{\ast}x^{\ast}$ for
every $x,y\in\mathcal{M}$). Let $j\in\mathcal{M}$ be such that
$j^{\ast}j=1$. Then
$\mathcal{G}=\set{x\in\mathcal{M}}{x^{\ast}jx=j}$ is a group.
Moreover, $\mathcal{G}=\set{x\in\mathcal{M}}{xjx^{\ast}=j}$.
\end{proposition}
\begin{proof}
Let $x,y\in\mathcal{G}$.
 Then $(xy)^{\ast}j(xy)=y^{\ast}(x^{\ast}jx)y=y^{\ast}jy=j$.
Hence $xy\in\mathcal{G}$ and $\mathcal{G}$ is closed under
multiplication. Further, since $j$ has a left inverse, it is
invertible, $jj^{\ast}=1$ and thus $1,j,j^{\ast}\in\mathcal{G}$. For
$x\in\mathcal{G}$ we have $x^{\ast}jx=j$, hence
$(j^{\ast}x^{\ast}j)x=1$. Thus $x$ is invertible,
$x^{-1}=j^{\ast}x^{\ast}j$ and $1=xx^{-1}=xj^{\ast}x^{\ast}j$. It
follows $j^{\ast}=xj^{\ast}x^{\ast}$ and applying the $\ast$
operation we get $j=xjx^{\ast}=(x^{\ast})^{\ast}jx^{\ast}$, since
$(x^{\ast})^{\ast}=x$. Finally $x^{\ast}\in\mathcal{G}$,
$x^{-1}=j^{\ast}x^{\ast}j\in\mathcal{G}$ and $\mathcal{G}$ is a
group. By a similar argument, $xjx^{\ast}=j$ implies $x^{\ast}jx=j$.
\end{proof}

\begin{corollary}\label{20.2}
$\Sp_{[n_{1},\dots,n_{k}]}$ is a finite subgroup of the monoid
$\mathcal{M}_{[n_{1},\dots,n_{k}]}$.
\end{corollary}

\begin{proposition}\label{23}
Let $H=(h_{ij})_{i,j=1,\dots,2k}\in\mathcal{M}_{[n_{1},\dots,n_{k}]}$,
$h_{ij}=\frac{n_{\lceil i/2\rceil}}{\gcd(n_{\lceil i/2\rceil},
n_{\lceil j/2\rceil})}a_{ij}$ and $a_{ij}\in\Z_{n_{\lceil i/2\rceil}}$
for $i,j=1,\dots,2k$. Then $H\in\Sp_{[n_{1},\dots,n_{k}]}$ if and only if
$$\sum^{k}_{m=1}\frac{n_{\lceil i/2\rceil}}{\gcd(n_{m},
n_{\lceil i/2\rceil})}\cdot\frac{n_{m}}{\gcd(n_{m},n_{\lceil
j/2\rceil})}
(a_{2m-1,i}a_{2m,j}-a_{2m-1,j}a_{2m,i})\equiv_{n_{\lceil i/2\rceil}}
w_{ij}$$ for every $i,j=1,\dots,2k$ (where
$J=(w_{ij})_{i,j=1,\dots,2k}\in\Sp_{[n_{1},\dots,n_{k}]}$).
\end{proposition}
\begin{proof}
We only transcribe the equation $H^{\ast}JH=J$ using
$h^{\ast}_{ij}=\frac{n_{\lceil i/2\rceil}}{\gcd(n_{\lceil
i/2\rceil},n_{\lceil j/2\rceil})}a_{ji}$, $w_{2m-1,2m}=1$,
$w_{2m,2m-1}=-1$ for $m=1,\dots,k$ and $w_{ij}=0$ otherwise.
\end{proof}

Due to \ref{19} the new groups $\Sp_{[n_{1},\dots,n_{k}]}$ represent
very specific generalization of symplectic groups over modular
rings, thus providing sufficient reason for our notation. Clearly,
for composite systems consisting of subsystems of equal dimensions
$n_{1}=\dots =n_{k}$, the new groups reduce to the well known
symplectic groups \cite{symplectic}:
\begin{corollary}\label{20.3}
If $n_1=\dots=n_k= n$, i.e. $N=n^k$, the symmetry group is
$\Sp_{[n,\dots,n]}\cong \Sp_{2k}(\Z_n)$.
\end{corollary}

These cases are of particular interest, since they uncover
symplectic symmetry of $k$-partite systems composed of subsystems
with the same dimensions. This circumstance was found, to our
knowledge, first in \cite{PST06} for $k=2$ under additional
assumption that $n=p$ is prime, leading to $\Sp_{4}(\mathrm{F}_p)$
over the field $\mathrm{F}_p$. We have generalized this result in
\cite{normalizer} to bipartite systems with arbitrary (non-prime)
$n=m$ yielding the symmetry group $\Sp_{4}(\mathbb{Z}_n)$ over the
modular ring $\mathbb{Z}_n$. The above corollary \ref{20.3} extends
this fact also to multipartite systems. Similar result has
independently been obtained in \cite{Han}, where symmetries of
tensored Pauli gradings of $\sl_{n^k}(\C)$ were investigated.

\section{Characterization of $\Sp_{[n_{1},\dots,n_{k}]}$}\label{sec4}

In this section we are going to prove theorem \ref{30} describing by
which elements the group $\Sp_{[n_{1},\dots,n_{k}]}$ is generated.
Let $n_{1},\dots,n_{k}\in\N$ be again fixed numbers.
\begin{definition}\label{24}
Let $\ell\in\Z$, $1\leq i<j\leq k$.
 We define special matrices
 $G_{ij}(\ell)\in\mathcal{M}_{[n_{1},\dots,n_{k}]}$
with $2\times 2$ blocks
 \begin{displaymath}\Big(G_{ij}(\ell)\Big)_{rs}:=
 \left\{\begin{array}{lll} I_{2} & \textrm{if $r=s$} \\
 \frac{n_{r}}{\gcd(n_{r},n_{s})}\ell\cdot
 \big(\begin{smallmatrix}0&0\\1&0\end{smallmatrix}\big) &
 \textrm{if $(r,s)=(i,j),(j,i)$}\\0 & \textrm{otherwise}\end{array}\right.
 \end{displaymath} where $r,s=1,\dots,k$.
\end{definition}

Further we note that
$$\SL_{2}(\Z_{n_{1}})\times\cdots\times\SL_{2}(\Z_{n_{k}})\cong$$
$$\cong\big\{\diag(H_{1},\dots,H_{k})\in\mathcal{M}_{[n_{1},\dots,n_{k}]}
\Big| \ \ H_{i}\in\M_{2}(\Z_{n_{i}})\ \&\ \det
H_{i}\equiv_{n_{i}}1\big\}.$$
 Thus we can assume
 $\SL_{2}(\Z_{n_{1}})\times\cdots\times\SL_{2}(\Z_{n_{k}})$
 to be naturally embedded into $\Sp_{[n_{1},\dots,n_{k}]}$.

\begin{lemma}\label{26}
$G_{ij}(\ell)=G_{ij}(1)^{\ell}$ for every $\ell\in\Z$ and $1\leq i<j\leq k$ and
 $G_{ij}(1)\in\Sp_{[n_{1},\dots,n_{k}]}$.
\end{lemma}
\begin{proof}
First consider a permutation $\pi$ of the set $\{1,\dots,k\}$. It
induces an isomorphism
$\varphi_{\pi}:\mathcal{M}_{[n_{1},\dots,n_{k}]}\to
\mathcal{M}_{[n_{\pi(1)},\dots,n_{\pi(k)}]}$. It is clear that
$H\in\Sp_{[n_{1},\dots,n_{k}]}$ if and only if
$\varphi_{\pi}(H)\in\Sp_{[n_{\pi(1)},\dots,n_{\pi(k)}]}$ for every
$H\in\mathcal{M}_{[n_{1},\dots,n_{k}]}$. Hence it is enough to show
our assertion for $G_{12}(\ell)$ only and this is equivalent to the
case $k=2$ which was already treated in \cite[A.4]{normalizer},
where $G_{12}(\ell)$ was denoted $ r(\ell)$.
\end{proof}

\begin{remark}\label{24.1}
Let $u=(a,b)^{T}\in\Z^{2}$.
Then there are $A,A'\in\SL_{2}(\Z)$ such that
$Au=(0,\gcd(a,b))^{T}$ and $A'u=(\gcd(a,b),0)^{T}$.

 \noindent We can assume $u\neq 0$. Then there are $k,l\in\Z$ such
that $ka+lb=\gcd(a,b)=:d$. Now just put
$A=\Big(\begin{smat2}{b/d}{-a/d}{k}{l}\end{smat2}\Big)$ and
$A'=J_{2}A$.
\end{remark}

Now let $\mathcal{G}$ denote the subgroup of
$\Sp_{[n_{1},\dots,n_{k}]}$ which is generated by
$\SL_{2}(\Z_{n_{1}})\times\cdots\times\SL_{2}(\Z_{n_{k}})$ and
$\set{G_{ij}(1)}{1\leq i<j\leq k}$. We are going to prove theorem
\ref{30} that $\mathcal{G}=\Sp_{[n_{1},\dots,n_{k}]}$. For this some
auxiliary notions are needed.

\begin{remark}\label{rem4.4}
\begin{enumerate}

\item\label{a1}  Consider the elements of
$\mathcal{S}_{[n_{1},\dots,n_{k}]}$ as $k\times k$ matrices of
$2\times 2$ blocks. Let $\Sigma_{k}$ be the set of all last (i.e.
the $k$-th) columns of the elements of
$\mathcal{S}_{[n_{1},\dots,n_{k}]}$ and, similarly, let
$\Sigma^{\ast}_{k}$ be the set of all last (i.e. the $k$-th) rows of
the elements of $\mathcal{S}_{[n_{1},\dots,n_{k}]}$.  Clearly, the
involution $\ast$ on $\mathcal{S}_{[n_{1},\dots,n_{k}]}$ induces a
bijection $\Sigma_{k}\to\Sigma^{\ast}_{k}$ (we will use the same
notation for it).

\item\label{b1} The congruence $\equiv$ on
$\mathcal{S}_{[n_{1},\dots,n_{k}]}$ induces naturally equivalences
on $\Sigma_{k}$ and $\Sigma^{\ast}_{k}$ (we will use again the same
notation for them and denote $[U]$ the equivalence class containing
an element $U$). Hence it easily follows that $U,U'\in\Sigma_{k}$,
$U\equiv U'$ and $H, H'\in\mathcal{S}_{[n_{1},\dots,n_{k}]}$,
$H\equiv H'$ imply $U^{\ast}\equiv(U')^{\ast}$ and $HU\equiv H'U'$.
Moreover, $(HU)^{\ast}=U^{\ast}H^{\ast}$.

\item\label{d1}  Now, put $\Omega_{k}:=\Sigma_{k}/_{\equiv}$
and $\Omega^{\ast}_{k}:=\Sigma^{\ast}_{k}/_{\equiv}$. By (\ref{a1}),
(\ref{b1}) and \ref{122}, we have a well defined map
$\Omega_{k}\to\Omega^{\ast}_{k}$ induced by $\ast$ and there is a
natural action (via the matrix multiplication) of  the ring
$\mathcal{M}_{[n_{1},\dots,n_{k}]}$  on the set $\Omega_{k}$.

\item\label{e1} Let $U,U'\in\Sigma_{k}$, $U\equiv U'$ and
$T, T'\in\Sigma^{\ast}_{k}$, $T\equiv T'$. Then $TU\equiv_{n_{k}}T'U'$.

  \noindent (Clearly, there are
$H,H'\in\mathcal{S}_{[n_{1},\dots,n_{k}]}$ such that $U$ ($U'$) is
the last column of $H$ ($H'$) and $H\equiv H'$. Similarly, there are
$G,G'\in\mathcal{S}_{[n_{1},\dots,n_{k}]}$ such that $T$ ($T'$) is
the last row of $G$ ($G'$) and $G\equiv G'$. Then $TU$ ($T'U'$) is
the block on the $(k,k)$-position of the matrix $GH$ ($G'H'$). By
\ref{11e} part (\ref{f}) we have $GH\equiv G'H'$ and thus
$TU\equiv_{n_{k}}T'U'$.)
\end{enumerate}
\end{remark}

Now we define the set
\begin{equation} \label{Delta}
\Delta_{k}:=\set{[U]\in\Omega_{k}}{U^{\ast}JU\equiv_{n_{k}}J_{2}}.
\end{equation}
It is by \ref{rem4.4} part (\ref{e1}) well defined. Using parts
(\ref{b1}) and (\ref{d1}) we see that $\Delta_{k}$ is invariant
under the action of the group $\Sp_{[n_{1},\dots,n_{k}]}$ (this
action is a restriction of the action of
$\mathcal{M}_{[n_{1},\dots,n_{k}]}$ on $\Omega_{k}$ that was
considered above).

\begin{proposition}\label{29}
$\mathcal{G}$ acts transitively on $\Delta_{k}$.
\end{proposition}
\begin{proof}
In this proof we will consider an element from $\Omega_{k}$ as an
ordered pair of its columns, i.e. as $(v,u)$ where
$v,u\in\mathcal{K}=\Z_{n_{1}}^{2}\times\cdots\times\Z_{n_{k}}^{2}$
are $2k$-tuples.  Note that for $v^{0}=(0,\dots,0,1,0)^T$ and
$u^{0}=(0,\dots,0,1)^T$ the pair $(v^{0},u^{0})$ belongs to
$\Delta_{k}$.

Now assume that some $(v,u)\in\Delta_{k}$ is given. To prove our
assertion, we construct for some $n\in\N$ a sequence of pairs ending
with $(v^{0},u^{0})$, i.e.
$(v,u)=(v_{0},u_{0}),\dots,(v_{n},u_{n})=(v^{0},u^{0})$ in
$\Delta_{k}$ and another sequence of matrices $H_{1},\dots,H_{n}$ in
$\mathcal{G}$ such that $(v_{j+1},u_{j+1})=H_{j+1}(v_{j},u_{j})$ for
$j=0,\dots,n-1$. We divide the proof into several steps.
  We put $d_{(i,j)}=\frac{n_{i}}{\gcd(n_{i},n_{j})}$ and note that
$d_{(i,i)}=1$.

(1) By \ref{24.1}, there are $B_{i}\in\SL_{2}(\Z_{n_{i}})$ for
$i=1,\dots,k$ such that for
$H_{1}:=\diag(B_{1},\dots,B_{k})\in\mathcal{G}$ we have
$$u_{1}:=H_{1}u=\big(d_{(1,k)}a_{1},0,\dots,d_{(k,k)}a_{k},0\big)^{T}$$
for some $a_{i}\in\Z_{n_{i}}$.

(2) Let
$$v_{1}:=\big(d_{(1,k)}b_{1},d_{(1,k)}c_{1},\dots,
d_{(k,k)}b_{k},d_{(k,k)}c_{k}\big)^{T}.$$
Then, by the definition of $\Delta_{k}$, we have
$\sum_{m=1}^{k}d_{(k,m)}d_{(m,k)}a_{m}c_{m}\equiv_{n_{k}}-1$. Put
$H_{2}:=\diag(I_{2},\dots, I_{2},B)\in\mathcal{G}$, where
$B:=\Big(\begin{smat2}{1}{0}{c_{k}}{1}\end{smat2}\Big)$. Then
 $$u_{2}:=H_{2}u_{1}=\big(d_{(1,k)}a_{1},0,\dots,
 d_{(k-1,k-1)}a_{k-1},0,a_{k},a_{k}c_{k}\big)^{T}.$$
Next by induction on $1\leq m\leq k-1$ we get that for
$H_{m+2}:=G_{mk}(c_{m})$ (where $G_{ij}(\ell)$ was defined in
\ref{24})
  $$u_{m+2}:=H_{m+2} u_{m+1}=$$
 $$=\big(\dots,\ d_{(m+1,k)}a_{m+1},0,\dots, d_{(k-1,k)}a_{k-1},0, a_{k},
  \Big(a_{k}c_{k}+\sum_{i=1}^{m}d_{(k,i)}d_{(i,k)}a_{i}c_{i}\Big)\big)^{T}.$$
 Thus $$u_{k+1}=\big(\dots, a_{k}, -1 \big)^{T}.$$

 (3) Using a similar argument as in step (1), we get that there is
$H_{k+2}\in\mathcal{G}$ such that
$$u_{k+2}:=H_{k+2}u_{k+1}=\big(0,d_{(1,k)}a'_{1},
\dots,0,d_{(k-1,k-1)}a'_{k-1}, 1,0 \big)^{T}$$ for some
$a'_{i}\in\Z_{n_{i}}$. Further put $H_{k+3}:=G_{1,k}(-a'_{1})\cdots
G_{k,k}(-a'_{k})\in\mathcal{G}$. Then, clearly,
$$u_{k+3}:=H_{k+3}u_{k+2}=\big(0,\dots,0,1,0\big)^{T}.$$

(4) Using again a similar argument as in step (1), we get that there
is $H_{k+4}\in\mathcal{G}$ such that
$$u_{k+4}:=H_{k+4}u_{k+3}=\big(0,\dots,0,1\big)^{T}$$ and
$$v_{k+4}:=\big(0,d_{(1,k)}b'_{1},\dots,0,
d_{(k-1,k)}b'_{k-1}, b',c'\big)^{T}$$ for some
$b'_{i}\in\Z_{n_{i}}$, $b',c'\in\Z_{n_{k}}$. Now we get from the
defining equation (\ref{Delta}) for $\Delta_{k}$ that
$b'\equiv_{n_{k}}1$. Put
$B':=\Big(\begin{smat2}{1}{0}{-c'}{1}\end{smat2}\Big)$. Then for
$H_{k+5}:=\diag(I_{2},\dots,I_{2},B')\in\mathcal{G}$ we get that
$$u_{k+5}:=H_{k+5}u_{k+4}=\big(0,\dots,0,1\big)^{T}$$
 and
$$v_{k+5}:=H_{k+5}v_{k+4}=
\big(0,d_{(1,k)}b'_{1},\dots, 0,d_{(k-1,k)}b'_{k-1}, 1,0\big)^{T}.$$

So we are in an analogous situation to step (3) and thus there is
$H_{k+6}\in\mathcal{G}$ such that
$$u_{k+6}:=H_{k+6}u_{k+5}=\big(0,\dots,0,1\big)^{T}$$ stays
unchanged  and
   $v_{k+6}:=H_{k+6}v_{k+5}=\big(0,\dots,0,1,0\big)^{T}$.
   \end{proof}

\begin{lemma}\label{29.1}
Let $H\in\mathcal{M}_{[n_{1},\dots,n_{k-1}]}$ and assume that
$T\in\Sigma^{\ast}_{k-1}$ is such that
$\left(\begin{smat2}{H}{0}{T}{I_{2}}\end{smat2}\right)
\in\Sp_{[n_{1},\dots,n_{k}]}$. Then $T=0$ and
$H\in\Sp_{[n_{1},\dots,n_{k-1}]}$.
\end{lemma}
\begin{proof}
There is $U\in\Sigma_{k}$ such that $T=U^{\ast}$. We have
$$\left(\begin{array}{cc}J&0\\0&J_{2}\end{array}\right)=
\left(\begin{array}{cc}H^{\ast}&U\\0&I_{2}\end{array}\right)
\left(\begin{array}{cc}J&0\\0&J_{2}\end{array}\right)
\left(\begin{array}{cc}H&0\\U^{\ast}&I_{2}\end{array}\right)=
\left(\begin{array}{cc}H^{\ast}JH+UJ_{2}
U^{\ast}&UJ_{2}\\J_{2}U^{\ast}&J_{2}\end{array}\right).$$
Hence $U^{\ast}=0$ and $H^{\ast}JH=J$.
\end{proof}

\begin{theorem}\label{30}
The group $\Sp_{[n_{1},\dots,n_{k}]}$ is generated by
$\SL_{2}(\Z_{n_{1}})\times\cdots\times\SL_{2}(\Z_{n_{k}})$
and $\set{G_{ij}(1)}{1\leq i<j\leq k}$.
\end{theorem}
\begin{proof}
Let $G\in\Sp_{[n_{1},\dots,n_{k}]}$ and $U\in\Sigma_{k}$ be the last
column of $G$.  Then $U\in\Delta_{k}$ by \ref{23}. Hence by \ref{29}
there is $G'\in\mathcal{G}$ such that
$G'G=\left(\begin{smat2}{H}{0}{T}{I_{2}}\end{smat2}\right)$ for some
$H\in\mathcal{M}_{[n_{1},\dots,n_{k-1}]}$ and
$T\in\Sigma^{\ast}_{k-1}$. Using \ref{29.1}, we have
$G'G=\left(\begin{smat2}{H}{0}{0}{I_{2}}\end{smat2}\right)$ with
$H\in\Sp_{[n_{1},\dots,n_{k-1}]}$. Now, by repeating this argument
several times, we find $\widetilde{G}\in\mathcal{G}$ such that
$\widetilde{G}G=I_{2k}$. Hence $G=\widetilde{G}^{-1}\in\mathcal{G}$
and we conclude with $\Sp_{[n_{1},\dots,n_{k}]}=\mathcal{G}$.
\end{proof}

\section{The normalizer of $\P_{(n_{1},\dots,n_{k})}$}\label{sec5}

In this section the normalizer is completely described and the main
theorem \ref{37} is proved. It contains our principal result that
the symmetry group, being the quotient of the normalizer, is indeed
isomorphic to $\Sp_{[n_{1},\dots,n_{k}]}$.

For proving the isomorphism between the group
$\mathcal{N}(\mathcal{P}_{(n_{1},\dots,n_{k})})/\mathcal{P}_{(n_{1},\dots,n_{k})}$
and $\Sp_{[n_{1},\dots,n_{k}]}$, we will consider elements of
$\mathcal{M}_{[n_{1},\dots,n_{k}]}$ as matrices $2k\times 2k$
instead of taking them as matrices $k\times k$ of blocks $2\times
2$, as we did so far. More precisely,
$H\in\Sp_{[n_{1},\dots,n_{k}]}$ will be treated as
$H=(h_{ij})_{i,j=1,\dots,2k}$, where
 $$
h_{ij}=\frac{n_{\lceil i/2\rceil}}{\gcd(n_{\lceil i/2\rceil},n_{\lceil j/2\rceil})}a_{ij}
$$
 for some $a_{ij}\in\Z_{n_{\lceil i/2\rceil}}$ and all $i,j=1,\dots,2k$.

\begin{definition}
Define $$\mathcal{N}(\mathcal{P}_{(n_{1},\dots,n_{k})}):=
N_{\Inn(\GL_{n_{1}\cdots n_{k}}(\C))}(\mathcal{P}_{(n_{1},\dots,n_{k})}),$$
 the normalizer of $\mathcal{P}_{(n_{1},\dots,n_{k})}$ in
 $\Inn(\GL_{n_{1}\cdots n_{k}}(\C))$.
 Further define
$$\mathcal{N}(\P_{n}):=N_{\Inn(\GL_{n}(\C))}(\P_{n}),$$
the normalizer of $\P_{n}$ in $\Inn(\GL_{n}(\C))$, and
 $$\mathcal{N}(\P_{n_{1}})\times\cdots\times\mathcal{N}(\P_{n_{k}}):=
 \set{\Ad_{M_{1}\otimes \cdots\otimes M_{k}}}{M_{i}\in\mathcal{N}(\P_{n_{i}})}
 \sub\Inn(\GL_{n_{1}\cdots n_{k}}(\C)).$$
\end{definition}

\begin{remark}\
\begin{enumerate}
\item Clearly,
$\mathcal{N}(\P_{n_{1}})\times\cdots\times\mathcal{N}(\P_{n_{k}})
\subseteq\mathcal{N}(\mathcal{P}_{(n_{1},\dots,n_{k})})$.

\item Consider now the usual natural homomorphism
$$\Psi:\mathcal{N}(\mathcal{P}_{(n_{1},\dots,n_{k})})\to
\Aut(\mathcal{P}_{(n_{1},\dots,n_{k})})$$
 given by
 $$\Psi(\Ad_{M})(\Ad_{X}):=\Ad_{M}\Ad_{X}\Ad_{M}^{-1}$$
 for every
 $\Ad_{M}\in\mathcal{N}(\mathcal{P}_{(n_{1},\dots,n_{k})})$ and
 $\Ad_{X}\in\mathcal{P}_{(n_{1},\dots,n_{k})}$.

\noindent We have $\ker(\Psi)=C_{\Inn(\GL_{n_{1}\dots
n_{k}}(\C))}(\mathcal{P}_{(n_{1},\dots,n_{k})})=\mathcal{P}_{(n_{1},\dots,n_{k})}$,
by lemma \ref{10}.

\item Further we put
$$\lambda_{ij}=\exp\Big(2\pi\mathrm{i}\frac{w_{ij}}{n_{\lceil i/2\rceil}}\Big)$$
for $i,j=1,\dots,2k$, where $w_{ij}$ are the entries of the matrix
$J\in\Sp_{[n_{1},\dots,n_{k}]}$ defined in \ref{19}. Using
(\ref{2A}) and (\ref{2PQ}) we can write the commutation relations
\begin{equation} \label{5PQ}
 A_{i}^{m}A_{j}^{n}=\lambda_{ij}^{mn}A_{j}^{n}A_{i}^{m}
 \end{equation}
 for all pairs $i,j=1,\dots,2k$ and $m,n\in\Z$.
\end{enumerate}
\end{remark}

\begin{lemma}\label{34}
$\Phi\Psi(\mathcal{N}(\mathcal{P}_{(n_{1},\dots,n_{k})}))\sub\Sp_{[n_{1},\dots,n_{k}]}$.
\end{lemma}
\begin{proof}
Let $\Ad_{G}\in\mathcal{N}(\mathcal{P}_{(n_{1},\dots,n_{k})})$,
where $G\in\GL_{n_{1}\cdots n_{k}}(\C)$. By \ref{31}, there is
$H=(h_{ij})_{i,j=1,\dots,2k}\in\mathcal{M}_{[n_{1},\dots,n_{k}]}$
such that $\Phi\Psi(\Ad_{G})=H$. Especially for $e_{j} =\Ad_{A_{j}}$
we have
$$\Ad_{GA_{j}G^{-1}}=\Psi(\Ad_{G})(e_{j})=
\prod\limits^{2k}_{i=1}e_{i}^{h_{ij}}=
\prod\limits^{2k}_{i=1}\Ad_{A_{i}^{h_{ij}}}.$$ So there are
constants $0\neq\nu_{j}\in\C$ such that
$$GA_{j}G^{-1}=\nu_{j}A_{1}^{h_{1,j}}\cdots A_{2k}^{h_{2k,j}}$$
for $j=1,\dots,2k$. Further,
$$GA_{i}A_{j}G^{-1}=
GA_{i}G^{-1}GA_{j}G^{-1}=
\nu_{i}\nu_{j}A_{1}^{h_{1,i}}\cdots A_{2k}^{h_{2k,i}}A_{1}^{h_{1,j}}
\cdots A_{2k}^{h_{2k,j}}=$$
$$=\nu_{i}\nu_{j}
\big(\prod\limits^{k}_{m=1}\lambda_{2m,2m-1}^{h_{2m,i}h_{2m-1,j}}\big)
A_{1}^{h_{1,i}+h_{1,j}}\cdots A_{2k}^{h_{2k,i}+h_{2k,j}}$$
using the commutation relations (\ref{5PQ}) (where the only
non-commuting elements are pairs $A_{2m-1}$, $A_{2m}$ for
$m=1,\dots,k$). On the other hand,
$$GA_{i}A_{j}G^{-1}=
\lambda_{ij}GA_{j}A_{i}G^{-1}=
\nu_{i}\nu_{j}\lambda_{ij}
\big(\prod\limits^{k}_{m=1}\lambda_{2m,2m-1}^{h_{2m,j}h_{2m-1,i}}\big)
A_{1}^{h_{1,i}+h_{1,j}}\cdots A_{2k}^{h_{2k,i}+h_{2k,j}}.$$
Thus
$$\prod\limits^{k}_{m=1}e^{-2\pi\mathrm{i}(h_{2m,i}h_{2m-1,j}/n_{m})}=
\lambda_{ij}\prod\limits^{k}_{m=1}e^{-2\pi\mathrm{i}(h_{2m,j}h_{2m-1,i}/n_{m})}$$
for every $i,j=1,\dots,2k$, i.e.
$$\exp\Big(2\pi\mathrm{i}\big(-\frac{w_{ij}}{n_{\lceil i/2\rceil}}+
\sum^{k}_{m=1}\frac{h_{2m-1,i}h_{2m,j}-h_{2m-1,j}h_{2m,i}}{n_{m}}\big)\Big)=1.$$
Since $h_{ij}=\frac{n_{\lceil i/2\rceil}}{\gcd(n_{\lceil i/2\rceil},n_{\lceil j/2\rceil})}a_{ij}$
for some $a_{ij}\in\Z_{n_{\lceil i/2\rceil}}$,
 by \ref{31} we get
$$-\frac{w_{ij}}{n_{\lceil i/2\rceil}}+
\sum^{k}_{m=1}\frac{n_{m}}{\gcd(n_{m},n_{\lceil
i/2\rceil})\gcd(n_{m},n_{\lceil j/2\rceil})}
(a_{2m-1,i}a_{2m,j}-a_{2m-1,j}a_{2m,i})\in\Z.$$
 This means that
$$\sum^{k}_{m=1}\frac{n_{\lceil i/2\rceil}}{\gcd(n_{m},n_{\lceil
i/2\rceil})}\cdot\frac{n_{m}}{\gcd(n_{m},n_{\lceil
j/2\rceil})}(a_{2m-1,i}a_{2m,j}-a_{2m-1,j}a_{2m,i})\equiv_{n_{\lceil
i/2\rceil}} w_{ij}$$ for every $i,j=1,\dots,2k$. Hence, by \ref{23},
$H\in\Sp_{[n_{1},\dots,n_{k}]}$.
\end{proof}

\begin{definition}
Let $1\leq i<j\leq k$. Put
$$T_{ij}=I_{n_{i+1}\cdots n_{j-1}}
\otimes Q_{n_{j}}^{\frac{n_{j}}{\gcd(n_{i},n_{j})}}$$
and
$$R_{ij}=I_{n_{1}\cdots n_{i-1}}
\otimes\mathrm{diag}(I_{n_{i+1}\cdots n_{j}},T_{ij},\dots,T_{ij}^{n_{i}-1})
\otimes I_{n_{j+1}\cdots n_{k}}.$$
\end{definition}

\begin{remark}\label{34.1}
For a ring $\mathcal{R}$,  $\M_{n}(\mathcal{R})$ is the ring of $n
\times n$ matrices with entries from $\mathcal{R}$. For
$a\in\mathcal{R}$ denote $Q_{[a]}:=\diag(1,a,a^{2},\dots,a^{n-1})\in
\M_{n}(\mathcal{R})$ and $P\in \M_{n}(\mathcal{R})$, where
$(P)_{i,j}:=\delta_{i,j-1}\cdot 1_{\mathcal{R}}$ for $i,j\in\Z_n$.
Let $E$ denote the identity matrix.
\begin{enumerate}
\item Let $a\in\mathcal{R}$ be such that $a^{n}=1$.
 Then $PQ_{[a]}=(aE)Q_{[a]}P$.
\item Let $a,b,\omega\in\mathcal{R}$ be such that $ab=\omega ba$.
 Then $Q_{[a]}(bE)=Q_{[\omega]}(bE)Q_{[a]}$.
\end{enumerate}
\end{remark}

\begin{lemma}\label{35}
Let $1\leq i<j\leq k$.
Then $\Ad_{R_{ij}}\in \mathcal{N}(\mathcal{P}_{(n_{1},\dots,n_{k})})$ and
$\Phi\Psi(\Ad_{R_{ij}})=G_{ij}(-1)\in\Sp_{[n_{1},\dots,n_{k}]}$.
\end{lemma}
\begin{proof}
$R_{ij}$ is a regular diagonal matrix, so are $A_{2m}$,
$m=1,\dots,k$, and thus these matrices commute. Further, for $m$
such that $1\leq m<i$ or $j<m\leq k$, the matrices $R_{ij}$ and
$A_{2m-1}$ also commute. Let now $m$ be such that $i<m<j$, then
 $$A_{2m-1}=I_{n_{1}\cdots n_{i-1}}
 \otimes\mathrm{diag}(U,U,\dots,U)\otimes I_{n_{j+1}\cdots n_{k}},$$
where $U=I_{n_{i+1}\cdots n_{m-1}}\otimes P_{n_{m}} \otimes
I_{n_{m+1}\cdots n_{j-1}}\otimes I_{n_{j}}$ and
$$R_{ij}=I_{n_{1}\cdots n_{i-1}}
\otimes\mathrm{diag}(V^{0},V^{1},\dots,V^{n_{i}-1}) \otimes
I_{n_{j+1}\cdots n_{k}},$$ where $V=I_{n_{i+1}\cdots n_{m-1}}\otimes
I_{n_{m}}\otimes I_{n_{m+1}\cdots n_{j-1}}\otimes
Q_{n_{j}}^{\frac{n_{j}}{\gcd(n_{i},n_{j})}}$. Hence $UV=VU$ and we
have the commutativity of $R_{ij}$ and $A_{2m-1}$ again.

  Now we use \ref{29} to treat the remaining cases. Put $n=n_{i}$,
  $\mathcal{R}=\M_{n_{i+1}\cdots n_{j}}(\C)$ and $a=T_{ij}$.
By \ref{34.1}(1), we have
$$\diag(T_{ij}^{0},T_{ij},T_{ij}^{2},\dots,T_{ij}^{n_{i}-1})(P_{n_{i}}
\otimes I_{n_{i+1}\cdots n_{j}})
\Big(\diag(T_{ij}^{0},T_{ij},T_{ij}^{2},\dots,T_{ij}^{n_{i}-1})\Big)^{-1}=$$
$$= Q_{[a]}PQ_{[a]}^{-1}=P(aE)^{-1}=(P_{n_{i}}\otimes
I_{n_{i+1}\cdots n_{j}})(I_{n_{i}}\otimes T_{ij})^{-1}.$$
 Tensoring this with $I_{n_{1}\cdots n_{i-1}}$ from the left and with
$I_{n_{j+1}\cdots n_{k}}$ from the right we get
$$R_{ij}A_{2i-1}R_{ij}^{-1}=A_{2i-1}A_{2j}^{-\frac{n_{j}}{\gcd(n_{i},n_{j})}}.$$

 Put $b=I_{n_{i+1}\cdots n_{j-1}}\otimes P_{n_{j}}$ and
 $\omega=e^{-2\pi\mathrm{i}/\gcd(n_{i},n_{j})}\cdot I_{n_{i+1}\cdots n_{j}}$.
Then $ab=\omega ba$ and by \ref{34.1}(2) we have
$$\diag(T_{ij}^{0},T_{ij},T_{ij}^{2},\dots,T_{ij}^{n_{i}-1})(I_{n_{i}\cdots
n_{j-1}}\otimes P_{n_{j}})
\Big(\diag(T_{ij}^{0},T_{ij},T_{ij}^{2},\dots,T_{ij}^{n_{i}-1})\Big)^{-1}=$$
$$=Q_{[a]}(bE)Q_{[a]}^{-1}=
Q_{[\omega]}(bE)=(Q_{n_{i}}\otimes I_{n_{i+1}\cdots
n_{j}})^{-\frac{n_{i}}{\gcd(n_{i},n_{j})}}(I_{n_{i}\cdots n_{j-1}}
\otimes P_{n_{j}})$$
 Tensoring this with $I_{n_{1}\cdots n_{i-1}}$ from the left and
with $I_{n_{j+1}\cdots n_{k}}$ from the right we get
$$R_{ij}A_{2j-1}R_{ij}^{-1}=A_{2i}^{-\frac{n_{i}}{\gcd(n_{i},n_{j})}}A_{2j-1}.$$
We conclude with $\Phi\Psi(\Ad_{R_{ij}})=G_{ij}(-1)$, where
$G_{ij}(\ell)$ was defined in \ref{24}.
\end{proof}

\begin{remark}\label{rem10}
Now the results obtained in \cite{HPPT02} have to be recalled since
they will be used in this and in the next section. They correspond
to the case $k=1$ with $n=n_{1}$. Using our notation we get that
$\Phi\Psi(\mathcal{N}(\P_{n}))=\SL_{2}(\Z_{n})$. Further, the group
$\SL_{2}(\Z_{n})$ is generated by
$\Big(\begin{smat2}{1}{1}{0}{1}\end{smat2}\Big)$ and
$\Big(\begin{smat2}{0}{-1}{1}{0}\end{smat2}\Big)$ and the group
$\mathcal{N}(\P_{n})$  is generated by $\Ad_{P_{n}}, \Ad_{Q_{n}},
\Ad_{D_{n}}$ and $\Ad_{S_{n}}$, where
$$(D_{n})_{ij}:=\delta_{ij}\varepsilon^{-i}\omega_{n}^{{i\choose
2}}$$ with $\varepsilon=\sqrt{-1}$ for $n$ even and $\varepsilon=1$
for $n$ odd, and
  $$(S_{n})_{ij}:=\omega_{n}^{ij}/\sqrt{n}.$$
It was shown in \cite{HPPT02} that
$\Phi\Psi(D_{n})=\Big(\begin{smat2}{1}{1}{0}{1}\end{smat2}\Big)$,
$\Phi\Psi(S_{n})=\Big(\begin{smat2}{0}{-1}{1}{0}\end{smat2}\Big)$
and $\ker(\Phi\Psi)=\P_{n}$.
\end{remark}
As an immediate consequence we have the following proposition.
\begin{proposition}\label{36}
$\Phi\Psi(\mathcal{N}(\P_{n_{1}})\times\cdots\times
\mathcal{N}(\P_{n_{k}}))=\SL_{2}(\Z_{n_{1}})\times\cdots\times\SL_{2}(\Z_{n_{k}})$.
\end{proposition}

Finally we arrive at the announced main theorems.
\begin{theorem}\label{37}

\begin{enumerate}
\item[(i)] $\mathcal{N}(\mathcal{P}_{(n_{1},\dots,n_{k})})/\mathcal{P}_{(n_{1},\dots,n_{k})}\cong\Sp_{[n_{1},\dots,n_{k}]}$
\item[(ii)] The group $\mathcal{N}(\mathcal{P}_{(n_{1},\dots,n_{k})})$ is generated by
$\mathcal{N}(\P_{n_{1}})\times\cdots\times\mathcal{N}(\P_{n_{k}})$ and\\ $\set{\Ad_{R_{ij}}}{1\leq i<j\leq k}$.
\end{enumerate}
\end{theorem}
\begin{proof}

(i) By \ref{30}, $\Sp_{[n_{1},\dots,n_{k}]}$ is generated by
$\set{G_{ij}(1)}{1\leq i<j\leq k}$ and $\SL_{2}(\Z_{n_{1}})\times
\cdots\times \SL_{2}(\Z_{n_{k}})$. Hence, by \ref{34}, \ref{35} and
\ref{36},
$\Phi\Psi(\mathcal{N}(\mathcal{P}_{(n_{1},\dots,n_{k})}))=\Sp_{[n_{1},\dots,n_{k}]}$.
Using \ref{31} and $\ker(\Psi)=\mathcal{P}_{(n_{1},\dots,n_{k})}$ we
get $\ker(\Phi\Psi)=\mathcal{P}_{(n_{1},\dots,n_{k})}$.

(ii) Let $\mathcal{N}$ be a     subgroup of $\mathcal{N}(\mathcal{P}_{(n_{1},\dots,n_{k})})$ generated by
$\mathcal{N}(\P_{n_{1}})\times\cdots\times\mathcal{N}(\P_{n_{k}})$ and $\set{\Ad_{R_{ij}}}{1\leq i<j\leq k}$. Then
$\ker(\Phi\Psi)=\mathcal{P}_{(n_{1},\dots,n_{k})}\sub\mathcal{N}(\P_{n_{1}})\times\cdots\times\mathcal{N}(\P_{n_{k}})\sub\mathcal{N}$
and, by \ref{35}, \ref{36} and \ref{30},
$\Phi\Psi(\mathcal{N})=\Sp_{[n_{1},\dots,n_{k}]}$. Hence
$\mathcal{N}=\mathcal{N}(\mathcal{P}_{(n_{1},\dots,n_{k})})$.
\end{proof}

\begin{theorem}\label{37.1}
There is a group
$\mathcal{G}_{(n_{1},\dots,n_{k})}\sub\U_{n_{1}\cdots n_{k}}(\C)$
such that
$\mathcal{N}(\mathcal{P}_{(n_{1},\dots,n_{k})})=
\set{\Ad_{M}}{M\in\mathcal{G}_{(n_{1},\dots,n_{k})}}$.
In particular, $\mathcal{G}_{(n_{1},\dots,n_{k})}$ is generated by
the matrices $$I_{n_{1}\cdots n_{i-1}}\otimes P_{n_{i}}\otimes
I_{n_{i+1}\cdots n_{k}}$$
 $$I_{n_{1}\cdots n_{i-1}}\otimes Q_{n_{i}}\otimes I_{n_{i+1}\cdots n_{k}}$$
  $$I_{n_{1}\cdots n_{i-1}}\otimes D_{n_{i}}\otimes I_{n_{i+1}\cdots n_{k}}$$
   $$I_{n_{1}\cdots n_{i-1}}\otimes S_{n_{i}}\otimes I_{n_{i+1}\cdots n_{k}}$$
 for $i=1,\dots,k$ and
 $$R_{ij}$$
 for $1\leq i<j\leq k$.
  \end{theorem}
  \begin{proof}
  Follows immediately from \ref{rem10} and \ref{37}.
  \end{proof}

\section{Mutually unbiased bases and the symmetry group}\label{sec6}

In this section we turn to the special cases described in corollary
\ref{20.3}, with $n_1=\dots=n_k= p$, where $p$ is a prime. For such
multipartite systems composed of subsystems with the same prime
dimension $p$ the Hilbert space is $\ell^{2}(\Z_p)\otimes \dots
\otimes \ell^{2}(\Z_p) \cong \ell^{2}(\Z_{p^{k}})$ and the symmetry
group is $\Sp_{[p,\dots,p]}\cong \Sp_{2k}(\Z_p)$. Our goal is to use
the symmetry group for an alternative proof of existence of a
maximal set of mutually unbiased bases (MUBs)\footnote{In the
Hilbert space $\ell^{2}(\Z_N)$ let $\{e_i \}_{i=1}^{N}$, $\{f_j
\}_{j=1}^{N}$ be two orthonormal bases. They are mutually unbiased,
if $\vert (e_i, f_j)\vert = 1/\sqrt{N}$ for all $i,j=1,\dots,N$. In
physical terms, if the system is in one of the states $e_i$, then
the probabilities to find the system in any of the states $f_j$ are
all equal to $1/N$.}
 in the Hilbert space of prime power dimension.

It is known that the number of mutually unbiased bases in a Hilbert
space of dimension $N$ must not exceed $N+1$
\cite{WoottersFields89}. It is also well known that the maximal
number $N+1$ is attained for $N$ being prime or power of a prime.
However, the determination of the maximal number of mutually
unbiased bases for other dimensions $N$ remains an open  problem as
yet.

Note that in this section the letter $k$ will be replaced by $n$,
thus an $n$-partite system is considered and the respective Hilbert
space $\ell^{2}(\Z_{p^{n}})$ is taken with the standard inner
product.

Essentially an idea of Bandyopadhyay, Boykin, Roychowdhury and Vatan
\cite{Bandyo} is used, but we provide a different proof. First we
recall their main point.

Denote
$$\Pi_{p}(n):=\set{M_{1}\otimes\cdots\otimes
M_{n}\in\GL_{p^{n}}(\C)}{M_{i}\in\Pi_{p}}.$$ For
$\alpha=(k_{1},\dots,k_{n},\ell_{1},\dots,\ell_{n})^{T}\in\Z_{p}^{2n}$
put $$A[\alpha]:=Q_{p}^{k_{1}}P_{p}^{\ell_{1}}\otimes\dots\otimes
Q_{p}^{k_{n}}P_{p}^{\ell_{n}}\in\Pi_{p}(n).$$
 For a $2n\times n$ matrix $U$ over $\Z_{p}$ assign a set of operators
$$\mathcal{C}(U):=\set{A[\alpha_{i}]}{i=1,\dots,n},$$ where
$\alpha_{i}$ is the $i$-th column of the matrix $U$.

MUBs are now constructed as orthonormal sets of common eigenvectors
of mutually commuting operators from $\mathcal{C}(U)$ for suitably
chosen $U$. Using the commutation relations for $P$ and $Q$ one
easily gets that
\begin{quote}
$A[\alpha]$ and $A[\beta]$ commute if and only if
$\alpha^{T}J'\beta=0$, where
$J':=\Big(\begin{smat2}{0}{-I_{n}}{I_{n}}{0}\end{smat2}\Big)$. Thus
$\mathcal{C}(U)$ consists of mutually commuting operators if and
only if $U^{T}J'U=0$.
 \end{quote}
  Now a special system $(\ast)$ of such
matrices fulfilling this condition is chosen, namely
$$
\big(\begin{smallmatrix}I_{n}\\ 0\end{smallmatrix}\big) \ \
\text{and} \ \ \big(\begin{smallmatrix}A_{i}\\
I_{n}\end{smallmatrix}\big) \ \ \text{for} \ \ i=1,\dots,p^{n}, \ \
\ \ \ \ \ \  (\ast)$$
 where $A_{i}\in\M_{n}(\Z_{p})$ are symmetric and
$A_{i}-A_{j}$ are regular for $i\neq j$. The existence of such a
system is shown further, see \ref{system}. In the following,
$\mathcal{C}(U)$ will always denote such a set of mutually commuting
operators.

We will now apply our previous results concerning the symmetry group
to get a different proof that the system $(\ast)$ indeed provides a
set of $p^{n}+1$ mutually unbiased bases. Moreover, we will show
that there is a group generating the MUBs from the canonical basis
via a natural action.

As already mentioned, we consider the case $n_{i}=p$ for each
$i=1,\dots,n$. Then the phase space
$\P_{(p,\dots,p)}\cong(\Z_{p})^{2n}$ is a vector space of dimension
$2n$ over $\Z_{p}$. In proposition \ref{31} we have considered the
homomorphism
$$\mathcal{N}(\mathcal{P}_{(p,\dots,p)})
\stackrel{\Phi}{\to}\End(\mathcal{P}_{(p,\dots,p)})
\cong\End(\Z_{p}^{2}\times\cdots\times\Z_{p}^{2})
\cong\mathcal{M}_{[p,\dots,p]}$$
 where the isomorphism
 $\End(\mathcal{P}_{(p,\dots,p)})\cong\mathcal{M}_{[p,\dots,p]}$
 was given with respect to the basis
 $(\Ad_{A_{1}},\dots,\Ad_{A_{2n}})$ of $\mathcal{P}_{(p,\dots,p)}$, where
$A_{2i-1}=I_{p^{i-1}}\otimes P_{p}\otimes I_{p^{n-i}}$ and
$A_{2i}=I_{p^{i-1}}\otimes Q_{p}\otimes I_{p^{n-i}}$ for
$i=1,\dots,n$.

Here it is useful to take a differently ordered basis, namely
$$(\Ad_{A_{2}},\Ad_{A_{4}},\dots,\Ad_{A_{2n}},
\Ad_{A_{1}},\Ad_{A_{3}},\dots,\Ad_{A_{2n-1}}).$$
 Then the corresponding automorphism of
$\mathcal{M}_{[p,\dots,p]}=\M_{2n}(\Z_{p})$,
 given by the above permutation matrix, transforms the symmetry group
$\Sp_{[p,\dots,p]}$ into the symplectic group over $\Z_{p}$
\cite{symplectic}
$$\Sp_{2n}(\Z_{p}):=\set{H\in\M_{2n}(\Z_{p})}{H^{T}J'H=J'}.$$
 Thus we can formulate our result as follows:
\begin{proposition}
There is a surjective homomorphism
$$\chi:\mathcal{N}(\mathcal{P}_{(p,\dots,p)})\to \Sp_{2n}(\Z_{p})$$
such that
$$\Ad_{M}\Ad_{A[\alpha]}\Ad_{M}^{-1}=\Ad_{A[\chi(\Ad_{M})\alpha]}$$
for every $\alpha\in\Z_{p}^{2n}$ and
$\Ad_{M}\in\mathcal{N}(\mathcal{P}_{(p,\dots,p)})$, where $M\in
\U_{p^{n}}(\C)$.
\end{proposition}
\begin{remark}\label{rem1}
Let $\Ad_{M}\in\mathcal{N}(\mathcal{P}_{(p,\dots,p)})$,
$U\in\Z_{p}^{2n\times n}$ and $\alpha_{i}$ be the $i$-th column of
$U$. Then the above property can be reformulated: for every
$i=1,\dots,n$ there is $0\neq\lambda_{i}\in\C$ such that $$M\cdot
A[\alpha_{i}] \cdot M^{-1}=\lambda_{i}A[\chi(\Ad_{M})\alpha_{i}].$$
Moreover, if $u\in\ell_{p^{n}}$ is a common eigenvector of the set
of operators $\mathcal{C}(U)$, then $Mu$ is a common eigenvector of
the set of operators $\mathcal{C}(\chi(\Ad_{M})U)$.
\end{remark}
\begin{proposition}\label{prop1}
Let $A,B\in\M_{n}(\Z_{p})$ be symmetric and $A-B$ be a regular matrix.
Then: \begin{enumerate}
\item[(i)] There is $H\in\Sp_{2n}(\Z_{p})$ such that
$H\big(\begin{smallmatrix}I_{n}\\ 0\end{smallmatrix}\big)=
\big(\begin{smallmatrix}I_{n}\\ 0\end{smallmatrix}\big)$ and
$H\big(\begin{smallmatrix}A\\I_{n}\end{smallmatrix}\big)=
\big(\begin{smallmatrix}0\\I_{n}\end{smallmatrix}\big)$.
\item[(ii)] There is $G\in\Sp_{2n}(\Z_{p})$ such that
$G\big(\begin{smallmatrix}A\\I_{n}\end{smallmatrix}\big)=
\big(\begin{smallmatrix}I_{n}\\0\end{smallmatrix}\big)$ and
$G\big(\begin{smallmatrix}B\\I_{n}\end{smallmatrix}\big)=
\big(\begin{smallmatrix}0\\ D\end{smallmatrix}\big)$
for some regular $D\in\M_{n}(\Z_{p})$.
\end{enumerate}
\end{proposition}
\begin{proof}
Note that for $A\in\M_{n}(\Z_{p})$,
$\Big(\begin{smat2}{I_{n}}{A}{0}{I_{n}}\end{smat2}\Big)$ belongs to
$\Sp_{2n}(\Z_{p})$ if and only if $A$ is symmetric. Put
$H=\Big(\begin{smat2}{I_{n}}{-A}{0}{I_{n}}\end{smat2}\Big)$ and
$G=\Big(\begin{smat2}{(A-B)^{-1}}{-(A-B)^{-1}B}{-I_{n}}{A}\end{smat2}\Big)$.
  \end{proof}
\begin{remark}\label{rem2}
 For $m\in\N$ the matrix
$S_{m}\in\M_{m}(\C)$ introduced in
 \ref{rem10} is unitary and induces the discrete Fourier transform
$$S_{m}Q_{m}S_{m}^{-1}=P_{m}\ \ \ \ \textrm{and}\ \ \ \
S_{m}P_{m}S_{m}^{-1}=Q_{m}^{-1}.$$
 Thus for
$\Ad_{S_{p}\otimes\cdots\otimes
S_{p}}\in\mathcal{N}(\mathcal{P}_{(p,\dots,p)})$ we have
$$\chi(\Ad_{S_{p}\otimes\cdots\otimes S_{p}})=J'\ \ \ \ \textrm{and}\ \ \ \
J'\big(\begin{smallmatrix}I_{n}\\
0\end{smallmatrix}\big)= \big(\begin{smallmatrix}0\\
I_{n}\end{smallmatrix}\big).
$$
\end{remark}
\begin{corollary}\label{cor1}
Let $U=\big(\begin{smallmatrix}I_{n}\\ 0\end{smallmatrix}\big)$ or
$U=\big(\begin{smallmatrix}A\\ I_{n}\end{smallmatrix}\big)$ where
$A\in\M_{n}(\Z_{p})$ is a symmetric matrix. Then there is an
orthonormal basis of common eigenvectors for the mutually commuting
operator set $\mathcal{C}(U)$.
\end{corollary}
\begin{proof}
For $U=\big(\begin{smallmatrix}I_{n}\\ 0\end{smallmatrix}\big)$
clearly the standard basis $\mathcal{E}$ is the desired basis. Let
$U=\big(\begin{smallmatrix}A\\ I_{n}\end{smallmatrix}\big)$, where
$A\in\M_{n}(\Z_{p})$ is a symmetric matrix. By \ref{prop1} (putting
e.g. $B=A-I_{n}$) there are $G\in\Sp_{2n}(\Z_{p})$ such that
$\big(\begin{smallmatrix}A\\I_{n}\end{smallmatrix}\big)=
G^{-1}\big(\begin{smallmatrix}I_{n}\\0\end{smallmatrix}\big)$
and $M\in\mathrm{U}_{p^{n}}(\C)$ such that
$\Ad_{M}\in\mathcal{N}(\mathcal{P}_{(p,\dots,p)})$ and
$\chi(\Ad_{M})=G$. Using \ref{rem1} and the unitarity of $M$, we
obtain the desired basis as $\set{M^{-1}e}{e\in\mathcal{E}}$.
 \end{proof}
\begin{proposition}\label{prop2}

(i) Let $D\in\M_{n}(\Z_{p})$ be regular and $\mathcal{B}$ be a basis
of common eigenvectors for $\mathcal{C}\big(\begin{smallmatrix}0\\
D\end{smallmatrix}\big)$. Then $\mathcal{B}$ is also a basis of
common eigenvectors for $\mathcal{C}\big(\begin{smallmatrix}0\\
I_{n}\end{smallmatrix}\big)$.

(ii) Let $\mathcal{B}$ be an orthonormal basis of common
eigenvectors for $\mathcal{C}\big(\begin{smallmatrix}I_{n}\\
0\end{smallmatrix}\big)$ and $u$ from $\mathcal{B}$. Then there is a
complex unit $\lambda$ such that $\lambda u$ belongs to the standard
basis $\mathcal{E}$ of $\C^{p^{n}}$.

(iii) Let $\mathcal{B}$ ($\mathcal{B}'$, respectively) be an
orthonormal basis of common eigenvectors for
 $\mathcal{C}\big(\begin{smallmatrix}I_{n}\\ 0\end{smallmatrix}\big)$
($\mathcal{C}\big(\begin{smallmatrix}0\\
I_{n}\end{smallmatrix}\big)$, respectively). Then $\mathcal{B}$ and
 $\mathcal{B}'$ are mutually unbiased.
 \end{proposition}
 \begin{proof}
(i) Since $D$ is invertible, we have that for every $i=1,\dots,n$,
$I_{p^{i-1}}\otimes P_{i}\otimes I_{p^{n-i}}$ lies in the group
generated by $\set{A[\alpha_{j}]}{j=1,\dots,n}$ where $\alpha_{j}$
is the $j$-the column of $\big(\begin{smallmatrix}0\\
D\end{smallmatrix}\big)$. Our assertion now follows immediately.

(ii) Let $\mathcal{E}_{p}$ be the standard basis of $\C^{p}$. Since
$u$ is an eigenvector of $Q_{p}\otimes I_{p^{n-1}}$ it is of the
form $u=e_{i_{1}}\otimes v$ for some $i_{1}=1,\dots,p$ and
$v\in\C^{p^{n-1}}$. Now, $u=e_{i_{1}}\otimes v$ is an eigenvector of
$I_{p}\otimes Q_{p}\otimes I_{p^{n-2}}$, hence it is of the form
$u=e_{i_{1}}\otimes e_{i_{2}}\otimes w$ for some $i_{2}=1,\dots,p$
and $w\in\C^{p^{n-2}}$.Repeating this argument we get that
$u=\lambda e_{i_{1}}\otimes \cdots\otimes e_{i_{n}}$ for
$i_{j}=1,\dots,p$ and $\lambda\in\C$. Since $u$ is normalized, it
follows that $|\lambda|= 1$.

(iii) Put $M=S_{p}\otimes\cdots\otimes
S_{p}\in\mathrm{U}_{p^{n}}(\C)$ (see \ref{rem2}). By \ref{rem1} and
\ref{rem2}, $M^{-1}\mathcal{B'}$ is an orthonormal basis of common
eigenvectors for $\mathcal{C}\big(\begin{smallmatrix}I_{n}\\
0\end{smallmatrix}\big)$. Hence, by (ii), there are matrices
$R_{1},R_{2}\in\GL_{p^{n}}(\C)$ with only one non-zero entry (a
complex unit) in each column and row, such that
$M^{-1}\mathcal{B'}=R_{1}\mathcal{E}$ and
$\mathcal{B}=R_{2}\mathcal{E}$. Now, let $u$ be from
$\mathcal{B}=R_{2}\mathcal{E}$ and $u'$ from
$\mathcal{B}'=MR_{1}\mathcal{E}$. Then there are $i,j\in\{1,\dots,
p^{n}\}$ such that $u=R_{2}e_{i}$ and $MR_{1}e_{j}$ with
$e_{i},e_{j}$ from $\mathcal{E}$. Hence
$$|(u,u')|=|(R_{2}e_{i},MR_{1}e_{j})|=|(R_{2}^{T}MR_{1})_{ij}|=
1/\sqrt{p^{n}},$$
 i.e. $\mathcal{B}$ and $\mathcal{B}'$ are mutually unbiased.
  \end{proof}

\begin{corollary}\label{cor2}
Let $U$ and $U'$ be distinct matrices from the system $(\ast)$, and
$\mathcal{B}$ and $\mathcal{B}'$ be orthonormal bases of common
eigenvectors for $\mathcal{C}(U)$ and $\mathcal{C}(U')$,
respectively. Then $\mathcal{B}$ and $\mathcal{B}'$ are mutually
unbiased.
 \end{corollary}
\begin{proof}
By \ref{prop1}, \ref{rem2} and \ref{37.1} there is
$M\in\mathrm{U}_{p^{n}}(\C)$ such that
$\chi(\Ad_{M})U=
\big(\begin{smallmatrix}I_{n}\\0\end{smallmatrix}\big)$
and $\chi(\Ad_{M})U'=
\big(\begin{smallmatrix}0\\ D\end{smallmatrix}\big)$
for some regular $D\in\M_{n}(\Z_{p})$.
By \ref{prop2}(i) and \ref{rem1}, $M\mathcal{B}$ (
$M\mathcal{B}'$, resp.) is
an orthonormal basis of common eigenvectors for
$\mathcal{C}\big(\begin{smallmatrix}I_{n}\\
0\end{smallmatrix}\big)$ ($\mathcal{C}\big(\begin{smallmatrix}0
\\ I_{n}\end{smallmatrix}\big)$, resp.).
Hence by \ref{prop2}(iii), the bases $M\mathcal{B}$ and
$M\mathcal{B}'$ are mutually unbiased. Finally, since $M$
is unitary, the bases $\mathcal{B}$ and $\mathcal{B}'$
are also mutually unbiased.
\end{proof}

Thus we have shown, using our knowledge of the normalizer of
$\mathcal{P}_{(p,\dots,p)}$, that having the system $(\ast)$, there
are $p^{n}+1$ mutually unbiased bases in the Hilbert space
$\ell^{2}(\Z_{p^{n}})$, where $p$ is a prime number. In the
remaining part of this section we show how to generate these bases
from the canonical one using one particular operator and an
elementary commutative group of order $p^{n}$ consisting of unitary
diagonal matrices, i.e. $\cong\Z_{p}^{n}$.

First, recall a result by Wootters and Fields in
\cite{WoottersFields89} (mentioned also in \cite{Bandyo}) that
supports the existence of a system $(\ast)$:
\begin{proposition}\label{system}
There are symmetric matrices $B_{1},\dots ,B_{n}\in\M_{n}(\Z_{p})$
such that for every $0\neq(\alpha_{1},\dots,\alpha_{n})^{T}\in\Z_{p}^{n}$
the matrix $\sum^{n}_{\ell=1}\alpha_{\ell}B_{\ell}$ is regular.
In particular, let $\gamma_{1},\dots , \gamma_{n}$ be a basis
of the finite field $\mathbb{F}_{p^{n}}$ as a vector space over
the field $\mathbb{Z}_{p}$.
Then any element $\gamma_{i}\gamma_{j}\in\mathbb{F}_{p^{n}}$ can be written
uniquely as
$$\gamma_{i}\gamma_{j} =\sum\limits^{n}_{\ell=1}b_{ij}^{\ell}\gamma_{\ell}$$
where $b_{ij}^{\ell}\in\mathbb{Z}_{p}$.
Now $(B_{\ell})_{ij}=b_{ij}^{\ell}$ are the required matrices.
\end{proposition}

Now let $\mathcal{D}$ denote the additive subgroup of
$\M_{n}(\Z_{p})$ generated by $B_{1},\dots,B_{n}$ from \ref{system}.
Clearly, $\mathcal{D}\cong\Z_{p}^{n}$ and it is easy to see that
$$\mathcal{H}:=
\Big\{\Big(\begin{smat2}{I_{n}}{B}{0}{I_{n}}\end{smat2}\Big)\Big|\
B\in\mathcal{D}\Big\}$$ is a (multiplicative) commutative subgroup
of $\Sp_{2n}(\Z_{p})$ that has a natural action (via matrix
multiplication) on the set
$$\big\{\big(\begin{smallmatrix}C\\I_{n}\end{smallmatrix}\big)\big|
\ C\in\mathcal{D}\big\}.$$

We consider now the system $(\ast)$ naturally as
$\big\{\big(\begin{smallmatrix}I_{n}\\0\end{smallmatrix}\big)\big\}
\cup\big\{\big(\begin{smallmatrix}C\\I_{n}\end{smallmatrix}\big)\big|
\ C\in\mathcal{D}\big\}$
with the mappings
\begin{equation*}
\big(\begin{smallmatrix}I_{n}\\0\end{smallmatrix}\big) \quad
 \stackrel{J'}{\longrightarrow} \quad
\big(\begin{smallmatrix}0\\I_{n}\end{smallmatrix}\big) \quad
 \stackrel{\Big(\begin{smat2}{I_{n}}{A}{0}{I_{n}}\end{smat2}\Big)}{\longrightarrow}
 \quad \big(\begin{smallmatrix}A\\I_{n}\end{smallmatrix}\big) \quad
 \stackrel{\Big(\begin{smat2}{I_{n}}{B-A}{0}{I_{n}}\end{smat2}\Big)}{\longrightarrow}
 \quad \big(\begin{smallmatrix}B\\I_{n}\end{smallmatrix}\big).
\end{equation*}
\begin{remark}\label{rem3}
The subspace of all symmetric matrices in $\M_{n}(\Z_{p})$ has a basis consisting of
\begin{itemize}
 \item matrices $E_{ij}$,
where $1\leq i<j\leq n$, which have the entry 1 at the positions
$(i,j)$ and $(j,i)$  and zeros otherwise, and
 \item matrices $E_{i}$, with $1\leq i\leq n$ with the only non-zero entry 1
on the position $(i,i)$.
 \end{itemize}
Now, put $F_{i}:=I_{p^{i-1}}\otimes D_{i}\otimes I_{p^{n-i}}$ for
$i=1,\dots,n$. Using \ref{35} and \ref{rem10} we get that
$\chi(\Ad_{R_{ij}^{-1}})=\Big(\begin{smat2}{I_{n}}{E_{ij}}{0}{I_{n}}\end{smat2}\Big)$
and
$\chi(\Ad_{F_{i}})=\Big(\begin{smat2}{I_{n}}{E_{i}}{0}{I_{n}}\end{smat2}\Big)$.
Thus taking unitary diagonal matrices in $\ell^{2}(\Z_{p^{n}})$,
 $$K_{\ell}:=\Big(\prod\limits^{n}_{i=1}F_{i}^{b_{ii}^{\ell}}\Big)
 \Big(\prod\limits_{1\leq i<j\leq n}R_{ij}^{-b_{ij}^{\ell}}\Big)
 \in\mathrm{U}_{p^{n}}(\C)$$ for $\ell=1,\dots,n$, we get
  $$\chi(\Ad_{K_{\ell}})=\Big(\begin{smat2}{I_{n}}{B_{\ell}}{0}{I_{n}}\end{smat2}\Big)$$
 \end{remark}

Let now $\mathcal{K}$ denote the multiplicative subgroup of
$\GL_{p^{n}}(\C)$ generated by $K_{1},\dots,K_{n}$. We have an
isomorphism $\mathcal{K}\rightarrow \mathcal{H} :
K\mapsto\chi(\Ad_{K})$. Hence $\mathcal{K}\cong\Z_{p}^{n}$. We can
now choose our set of $p^{n}+1$ mutually unbiased bases in
$\ell^{2}(\Z_{p^{n}})$ as
$$\big\{\mathcal{E}\big\}\cup\big\{KS\mathcal{E}\big|\ K\in\mathcal{K}\big\}.$$

Indeed, by \ref{rem1}, \ref{rem2} we have that $\mathcal{E}$,
$\mathcal{S\mathcal{E}}$ and $KS\mathcal{E}$ are (in this order)
orthonormal bases of common eigenvectors for
$\mathcal{C}\big(\begin{smallmatrix}I_{n}\\0\end{smallmatrix}\big)$,
$\mathcal{C}\big(\begin{smallmatrix}0\\ I_{n}\end{smallmatrix}\big)$
and $\mathcal{C}\big(\begin{smallmatrix}A\\
I_{n}\end{smallmatrix}\big)$, where
$\chi(\Ad_{K})=\Big(\begin{smat2}{I_{n}}{A}{0}{I_{n}}\end{smat2}\Big)$.
Now, by \ref{cor2}, these bases are mutually unbiased. The group
$\mathcal{K}$ acts on the system of bases as follows:
\begin{equation*}
\mathcal{E} \quad \stackrel{S}{\longrightarrow} \quad
 \quad S\mathcal{E} \quad \stackrel{K}{\longrightarrow}
 \quad KS\mathcal{E}.
\end{equation*}
\begin{remark}\label{rem4}
To have a better insight into the matrices $S$ and
$K_{1},\dots,K_{n}$ we will express the numbering of columns and
rows as $p$-adic numbers, i.e. as $n$-tuples
$\alpha_{1}\dots\alpha_{n}$, with $\alpha_{i}\in\{0,\dots,p-1\}$,
that correspond to $\alpha_{1}p^{n-1}+\cdots+\alpha_{n}p^{0}$. In
this notation we get
$$S_{\alpha_{1}\dots\alpha_{n}, \beta_{1}\dots\beta_{n}}=
\omega_{p}^{\sum_{i}\alpha_{i}\beta_{i}}/\sqrt{p^{n}},$$
$$(F_{i})_{\alpha_{1}\dots\alpha_{n},\alpha_{1}\dots\alpha_{n}}=
\varepsilon^{-\alpha_{i}}\omega_{p}^{{\alpha_{i}\choose 2}},$$
$$(R_{ij})_{\alpha_{1}\dots\alpha_{n},\alpha_{1}\dots\alpha_{n}}=
\omega_{p}^{\alpha_{i}\alpha_{j}},$$ and
$$(K_{\ell})_{\alpha_{1}\dots\alpha_{n},\alpha_{1}\dots\alpha_{n}}=
\varepsilon^{-\sum_{i}b^{\ell}_{ii}\alpha_{i}}\cdot
\omega_{p}^{\sum_{i}b^{\ell}_{ii}{\alpha_{i}\choose 2}
-\sum_{i<j}b^{\ell}_{ij}\alpha_{i}\alpha_{j}}$$
where $i,j,\ell=1,\dots,n$, $i< j$ and $\varepsilon=\sqrt{-1}$
for $p=2$ and $\varepsilon=1$ otherwise.
\end{remark}

\section{Conclusions}

In this paper we have described the symmetry groups of finite
Heisenberg groups of arbitrary quantum systems consisting of a
finite number $k$ of subsystems with Hilbert spaces of finite
dimensions $n_{1},\dots,n_{k}$, thus extending our results obtained
for bipartite systems \cite{normalizer}. For such a finitely
composed quantum system the finite Heisenberg group is embedded in
$\GL_{N}(\C)$, $N=n_{1}\dots n_{k}$. It induces --- via inner
automorphisms $\Ad_M$ --- an Abelian subgroup
$\P_{(n_{1},\dots,n_{k})}$ in $\Inn (\GL_{N}(\C))$. We have studied
the normalizer of this Abelian subgroup in $\Inn (\GL_{N}(\C))$ and
have thoroughly described it. The obtained symmetry group
$\Sp_{[n_{1},\dots,n_{k}]}$ is the quotient group of the normalizer
(theorem \ref{37}) and its further characterization was given in
section \ref{sec4}.

The symmetry groups uncover deeper structure of FDQM. For instance,
the cases when $n_1=\dots = n_k = n$, $n\in \Z$, corresponding to
dimensions $N=n^k$, are of particular interest. Then the symmetry
group for a multipartite system with this special composition is
$\Sp_{[n,\dots,n]}\cong \mathrm{Sp}_{2k}(\Z_n)$, which extends the
bipartite case $\mathrm{Sp}_{4}(\Z_n)$ considered in
\cite{normalizer} and \cite{PST06}. Thus our class of symmetry
groups can be viewed as a very specific generalization of the
familiar symplectic groups over modular rings \cite{symplectic}.

We have exploited the cases when $n_1=\dots = n_k = p$, $p$ prime,
corresponding to prime power dimension $N=p^k$, in section
\ref{sec6}, where the symmetry group $\mathrm{Sp}_{2k}(\Z_p)$ is
applied to an alternative derivation of the maximal set of mutually
unbiased bases in Hilbert spaces of prime power dimensions. Our
group theoretic derivation uses the idea of \cite{SulcTolar07},
where a constructive existence proof for $k=1$, $N=p$ prime, was
based on consistent use of the symmetry group $\mathrm{Sp}_{2}(\Z_p)
\cong \mathrm{SL}_{2}(\Z_p)$.

Our motivation to study symmetries of finite Heisenberg groups not
in prime or prime power dimensions as in \cite{Balian,Vourdas07,
GHW04}, but for arbitrary dimensions stems from our previous
research where we obtained results not restricted to finite fields.
Especially recall our paper \cite{TolarChadz} on Feynman's path
integral and mutually unbiased bases. Also the recent paper
\cite{VourdasBanderier} belongs to this direction, by dealing with
quantum tomography over modular rings. The papers
\cite{DigernesVV,Digernes} support our motivation, too, since they
show that finite quantum mechanics with growing odd dimensions
yields surprisingly good approximations of ordinary quantum
mechanics on the real line. This suggests a promising subject of
research to extend the results of \cite{Balian,Neuhauser} on
$\SL_{2}(\mathrm{F}_q)$ from finite fields to modular rings.

\section*{Acknowledgements}

The first author (M.K.) was supported by the project LC505 of Eduard
\v Cech's Center for Algebra and Geometry. The second author (J.T.)
acknowledges partial support by the Ministry of Education of Czech
Republic, projects MSM6840770039 and LC06002.

\end{document}